%% file: main_process_functions.tex
\documentclass[a4paper,onecolumn,11pt,unpublished]{quantumarticle}
\pdfoutput=1

\input{preamble.tex}

\allowdisplaybreaks
\usepackage{mathtools}
\usepackage[utf8]{inputenc}
\usepackage[english]{babel}
\usepackage[T1]{fontenc}
\usepackage{comment}
\usepackage{subcaption}

\newtheorem{theorem}{Theorem}
\newtheorem{lemma}{Lemma}

\theoremstyle{definition}
\newtheorem{definition}{Definition}
\newtheorem{example}{Example}

\theoremstyle{remark}

\pgfsetlayers{background,strings,edgelayer,nodelayer,foreground,main}

\begin{document}

\title{Process matrices on generalised theories}

\author{Matt Wilson}
\author{James Hefford}
\author{Timothée Hoffreumon}
\affiliation{Université Paris-Saclay, CNRS, ENS Paris-Saclay, Inria, CentraleSupélec, Laboratoire Méthodes Formelles}

\maketitle

\begin{abstract}
We extend the process-matrix framework for defining indefinite causal orders to a broad class of generalised physical theories. 
We give a simple re-axiomatisation of spacetime contexts in terms of naturality over classical probes, which precisely captures non-contextual assignment together with compatibility with convexity and coarse-graining within local laboratories. 
Using this compositional perspective, we give a clean characterisation of theories in which abstract spacetime contexts admit a process-matrix representation, and we pin down the precise conditions (satisfied by classical theory, quantum theory, and a Boxworld-like theory) under which the supermap and process-matrix perspectives on higher-order processes coincide. 
\end{abstract}



Inspired by the development of the circuit formalism for quantum computation \cite{deutsch1989quantum} and Bell's theorem \cite{bells_thm}, two core concepts within modern quantum foundations are the compositional structure of systems and processes, and the observable correlations that such systems and processes can produce.
This can be seen both in the prominence of the study of generalised (and operational) probabilistic theories \cite{chiribella_informational, barrett_gpts}, which accommodates physical theories with stronger-than-quantum non-local correlations whilst retaining the circuit formalism, and in the study of indefinite causal order \cite{hardy2005probabilitytheoriesdynamiccausal, chiribella_switch, oreshkov, costa2026indefinitequantumcausality} in terms of higher-order quantum operations \cite{taranto2025higherorderquantumoperations}, which accommodates causal correlations without a definite causal order, such as those given by quantum superpositions of circuit structures, whilst retaining local compatibility with quantum theory. 


For beyond-quantum systems, it is possible, for instance, to remain compatible with the circuit formalism (and the basic no-superluminal signalling principle of special relativity) whilst producing even stronger correlations than those which arise from quantum entanglement using general no-signalling boxes \cite{Popescu:1994aa}. For compatibility with local quantum theory beyond the circuit framework, it is possible to find physically consistent arrangements of events termed multipartite causal correlations \cite{PhysRevA.94.032131} which violate a generalisation of Bell inequalities to causal inequalities, such as the OCB process \cite{oreshkov}, the Lugano process \cite{baumeler_incompatibility}, and the Grenoble process \cite{Wechs_2021}. These two generalisations lead quite naturally to a foundational question: what are the most general causal correlations within conceivable physical theories?

The problem of moving the study of indefinite causal order beyond classical and quantum theory has recently gained some attention.
A number of very recent works aim to directly put forward either the definition of higher-order processes on specific physical theories such as Boxworld \cite{bavaresco2024indefinitecausalorderboxworld, chen2024causalitydualitymultipartitegeneralized}, Real quantum theory \cite{surace2026indefinitecausalorderreverses}, and Time-Symmetric quantum theory \cite{mrini2024indefinitecausalstructurecausal}, or to focus directly on specific candidate higher-order processes on specific physical theories such as Hex-Square theory \cite{sengupta2024achievingmaximalcausalindefiniteness}.
These recent works are complementary to the development of more general frameworks for defining higher-order processes over broad classes of theories, such as the construction of higher-order causal categories \cite{kissinger_caus,simmons_completelogic,SimmonsKissinger2022} from pre-causal categories, and the more general construction of categorical supermaps \cite{wilson_locality, hefford_supermaps, hefford2025bvcategoryspacetimeinterventions} from general circuit models. When the circuit model at hand is pre-causal, the two approaches coincide \cite{wilson2026higherorderquantumobjectsstrong}, and for non-pre-causal models such as Boxworld, the adaptation of the Yoneda lemma \cite{maclane1998categories} to categorical supermaps \cite{wilson2026supermapsgeneralisedtheories} relates categorical supermaps to concrete approaches in the literature.

These general frameworks have however, aside from a recent minimal correlational principle put forward within broader classes of theories \cite{liu2025parityerasurefoundationalprinciple}, generally focused on the first of two alternative perspectives and respective axiomatisations for higher-order quantum processes. The former perspective is based on the completeness principle for transformations, in this context referred to as quantum supermaps \cite{chiribella_supermaps}. The latter is instead based on non-contextual assignment from instruments to probability distributions, in this context referred to as process matrices \cite{oreshkov}. 

This leaves open a number of foundational questions concerning the limits of causation beyond quantum mechanics: \textit{What is a process functional over a general physical theory? Under what circumstances do the process functionals on a theory admit an analogue of the process matrix representation? Under what conditions do supermaps and process matrices on a generalised theory coincide?}

In this paper, we address these questions in order to better understand the structure of causality across the spectrum of conceivable physical theories.

\paragraph{Results}
In this paper we develop a definition for process functionals on general theories (non-contextual assignments from instruments to probability distributions respecting normalisation, coarse-graining, and convexity), and prove equivalence to local-applicability (strong naturality) with respect to classical-variable environments. 
\begin{theorem}
On any generalised theory $\cat{C}$:
\[ \{\textrm{Process functionals} \} \cong \{ \textrm{Strong Natural Transformations on Classical Variables} \}. \]
\end{theorem}
This forms a direct analogue to the categorical supermaps program, which axiomatises supermaps through local applicability (strong naturality) with respect to general physical systems. 
\begin{figure}[htbp]
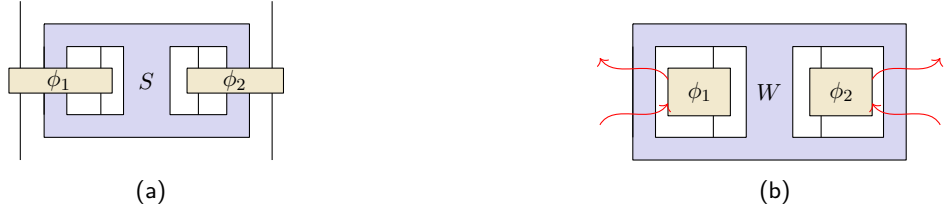

\centering

\begin{subfigure}{0.45\textwidth}
\centering
 \input{figs/intro_compare.tikz}
\caption{}
\label{fig:intro-compare-a}
\end{subfigure}
\hfill
\begin{subfigure}{0.45\textwidth}
\centering
 \input{figs/intro_compare_2.tikz}
\caption{}
\label{fig:intro-compare-b}
\end{subfigure}

\caption{Two alternative perspectives on higher-order operations: (a) a supermap is a higher-order process that can be locally applied to parts of lower-order processes; (b) a process matrix is an assignment that can be locally applied to the physical systems of instruments, treating the classical systems of an instrument as the outside environment.}
\label{fig:intro-perspectives}
\end{figure}

We then focus our attention on generalised theories with channel-state duality (an analogue of the Choi-Jamiolkowski isomorphism) \cite{Choi_1975, Jamiolkowski_1972}, characterising precisely those theories on which abstract process functionals support a process-matrix representation. More precisely, we show that a generalised theory supports a natural interpretation of the Busch-Gleason axiom if and only if there is a process-matrix representation theorem for multipartite process functionals on that theory. 
\begin{theorem}
Consider any generalised theory $\cat{C}$ with channel-state duality, the following are equivalent:
\begin{itemize}
\item The Busch-Gleason axiom
\item Every multipartite abstract process functional has a process-matrix representation:  \[   \input{figs/process_thm_1.tikz}  \ = \   \input{figs/process_thm_2.tikz}    \]
\end{itemize}
\end{theorem}

Finally, we prove that a generalised theory is furthermore tomographically closed (meaning that every process which produces normalised conditional distributions under tomography is a deterministic process) if and only if supermaps are equivalent to process matrices on that theory:
\begin{theorem}
Consider a generalised theory $\cat{C}$ with channel-state duality:
\[ \textrm{ $\cat{C}$ satisfies Busch-Gleason Axiom + Tomographic Closure}  \]
\[  \iff  \]
\[  \{ \textrm{Multipartite Process Functionals} \}  \cong  \{ \textrm{Multi-input Supermaps} \}  \]
\end{theorem}
We further discuss how these results apply to well-studied theories in the literature, namely, classical theory, quantum theory, and a Boxworld-like theory.

\section{Process functionals on general theories}
In the following, we frame our discussion within the context of generalised theories \cite{wilson2026supermapsgeneralisedtheories}: theories which build on the process-theoretic formalism \cite{coecke_kissinger_2017} supporting sequential composition and parallel composition, by adding an interface to the classical world in a similar vein to \cite{gogioso_cpt} and \cite{Selby2021reconstructing}. In this sense, generalised theories are an extension to operational probabilistic theories \cite{chiribella_purification} in which classical variables are given an explicit compositional interpretation in terms of wires.
\begin{definition}[Generalised Theories]
    A generalised theory is a convex enriched \cite{kelly, fritz2015convex} symmetric monoidal category $\cat{C}$, representing the non-deterministic physical events, and equipped with:
    \begin{itemize}
        \item a specified wide convex subcategory $\cat{C}^{\mathbf{d}} \hookrightarrow  \cat{C}$, representing the deterministic physical events,
        \item a full and faithful embedding $\stoch \hookrightarrow \cat{C}^{\mathbf{d}}$,
        \item a full and faithful embedding $\mat{\reals^+} \hookrightarrow \cat{C}$,
    \end{itemize}
    such that the following diagram commutes:
    \xsavebox{cground}{$\cat{C}^{\mathbf{d}}$}
    \[
    \begin{tikzcd}
        \stoch \arrow[r] \arrow[d] & \xusebox{cground} \arrow[d] \\
        \mat{\reals^+} \arrow[r] & \cat{C}
    \end{tikzcd}
    \]
    and such that all processes can be classically controlled. That is, for any family $\{M_{i}: A\rightarrow A'\}_i$ of deterministic processes (meaning the morphisms of $\cat{C}^{\mathbf{d}}$), there exists a deterministic physical process $M: A\otimes \mathbb{I}  \rightarrow A'$ and a deterministic classical process $i: I \rightarrow \mathbb{I}$ such that
    \begin{equation}\label{eq:control}
    	M \circ (1_A\otimes i)  = M_i \:.
    \end{equation}
\end{definition}
In practice, generalised theories are simply theories which support a well-behaved classical interface. As a practical guide to how such theories can be used, we show a few pictures that can be written within generalised theories and give their operational interpretation. 
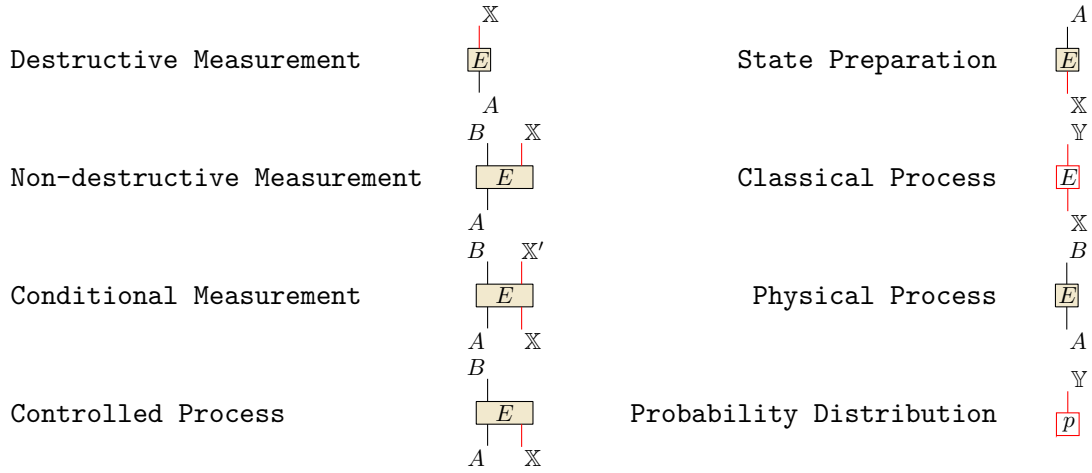
\begin{align*}
& \texttt{Destructive Measurement } && \input{figs/outline_1.tikz}  \quad  &&&  \texttt{State Preparation} &&&& \input{figs/outline_2.tikz} \\
& \texttt{Non-destructive Measurement } && \input{figs/outline_3.tikz}  \quad  &&&  \texttt{Classical Process} &&&& \input{figs/outline_7.tikz} \\
& \texttt{Conditional Measurement} && \input{figs/outline_5.tikz}  \quad  &&&  \texttt{Physical Process} &&&& \input{figs/outline_6.tikz} \\
& \texttt{Controlled Process} && \input{figs/outline_4.tikz}  \quad  &&&   \texttt{Probability Distribution} &&&& \input{figs/outline_8.tikz} 
\end{align*}
Convex enrichment also allows one to write convex and general additive combinations of diagrams.
Finally, we can refer to the $M_x^y$ component of a conditional non-destructive measurement (which we will generally refer to as a parameterised instrument) by simply writing
\[  M_{x}^y \ = \ \input{figs/param_inst2.tikz}.  \]
This represents choosing instrument $x$ and post-selecting on outcome $y$.

The canonical example of a generalised theory is the theory of quantum instruments, which compositionally can be expressed as the theory of quantum-classical processes. One can also consider, as a more algebraic example, the theory of quantum channels between $C^*$-algebras.
\begin{example}[Classical theory]
In classical theory we simply take $\cat{C}^{\mathbf{d}}  = \stoch$ and $\cat{C} = \mat{\reals^+}$. A deterministic process $f:n_A \rightarrow n_B$ is an $n_B \times n_A$ stochastic matrix, a deterministic state $\rho : I \rightarrow A$ is an $n_A$-dimensional probability vector, and a deterministic effect $E: A \rightarrow I$ is the transpose of the $n_A$ vector containing $1$'s in every entry.
\end{example}
\begin{example}[Quantum theory]
In quantum theory, we take $\cat{C}^{\mathbf{d}} $ to be the category of hybrid classical-quantum channels, that is, a morphism of type $(A   , n ) \rightarrow (B , m) $ with $n, m$ natural numbers is taken to be a function from elements of $[1, \dots ,n]$ to $m$-component instruments. More precisely, a morphism of type $(A,  n) \rightarrow (B, m) $ is taken to be a family $\{M_x^{y}\}_{x,y}$ of completely positive maps $M_x^y:A\rightarrow B$ such that for each fixed $x$ (often called a \textit{setting}) the family $\{M_x^{y}\}_{y}$ is a quantum instrument, meaning that $\sum_y M^y_x$ is a CPTP map. 
We take $\cat{C}$ to be the category of hybrid quantum-classical maps, meaning families $\{M_x^{y}\}_{x,y}$ of completely positive maps without any trace-preservation requirement. 
Sequential composition of $M: (A ,   n) \rightarrow (B,  m)$ with $N: (B,  m) \rightarrow (C ,  q)$ is given by the family $(N \circ M)_x^z = \sum_y N_y^{z} \circ M_x^{y}$. 
The parallel composition of systems is given by $(A ,  n) \otimes (B ,  m) := (A \otimes B ,n m)$.
The parallel composition of hybrid classical-quantum maps is given by $(N \otimes M)_{x,x'}^{y,y'} = N_x^{y} \otimes M_{x'}^{y'}$. The monoidal unit is given by $(\mathbb{C} ,1)$.
Finally, the classical systems are given by those of the form $(\mathbb{I} ,n)$, since morphisms of type $(\mathbb{I} ,n) \rightarrow (\mathbb{I} ,m)$ are then matrices of scalars $M_x^y :  \mathbb{I} \rightarrow \mathbb{I}$, thus matrices valued (up to isomorphism) in the non-negative reals. In the deterministic case, these matrices come with the extra requirement that $\sum_y M_x^y = 1$, which indeed is the requirement that the matrix be stochastic. 
\end{example}
\begin{example}[The theory of classical boxes]
The theory of classical boxes, termed $\mathbf{Box}$, is a Boxworld-like theory exhibiting the most general possible bipartite correlations. Whilst in previous works \cite{bavaresco2024indefinitecausalorderboxworld, wilson2026supermapsgeneralisedtheories}, including one by the authors of this paper, the theory $\mathbf{Box}$ has been literally referred to as \textit{Boxworld} - we choose here to avoid that term in order to disambiguate from the traditional definition of Boxworld in the literature \cite{barrett_gpts}. In $\mathbf{Box}$, objects are taken to be the symbols $\otimes_i (a_i \rightarrow a_i')$ with $a_i, a_i'$ positive integers, the morphisms of type $\otimes_i (a_i \rightarrow a_i') \rightarrow  \otimes_i (b_i \rightarrow b_i')$ are taken to be the real non-negative tensors (equivalently matrices, but referred to as tensors to disambiguate from our interpretation of states and matrices) $T_{b_1 \dots b_m a_1' \dots a_n'  }^{b_1'\dots b_m' a_1 \dots a_n}$. The deterministic morphisms are taken to be those tensors which furthermore preserve the no-signalling channels, meaning that, the contraction $T_{b_1 \dots b_m a_1' \dots a_n'  }^{b_1'\dots b_m' a_1 \dots a_n} M_{a_1 \dots a_n}^{a_1' \dots a_n'}$ with any non-signalling channel $M_{a_1 \dots a_n}^{a_1' \dots a_n'}$ returns a new non-signalling channel (i.e., an element of $\otimes_i (b_i \rightarrow b_i')$). It is easy to see from this definition that the states of this theory are indeed isomorphic to the classical non-signalling channels. The embedding of classical theory into Boxworld then works as follows, each classical system $a$ is identified with the Boxworld system $(1 \rightarrow a)$. It is easy to check that the non-deterministic morphisms $(1 \rightarrow a) \rightarrow (1 \rightarrow b)$ are automatically positive real matrices of type $a \rightarrow b$ and that the deterministic morphisms of type $(1 \rightarrow a) \rightarrow (1 \rightarrow b)$ are precisely the stochastic matrices of type $a \rightarrow b$. 
\end{example}

A signature of finite dimensionality in generalised theories is the existence of compact closure. A compact closed category is a symmetric monoidal category equipped with cups and caps \cite{abramsky_coecke}, that is, preferred processes $\cup: I \rightarrow A \otimes A^*$ and $\cap: A^* \otimes A \rightarrow I$ satisfying the snake equation \[  \input{figs/snake_1.tikz} =  \input{figs/snake_2.tikz}  . \] We will be interested in generalised theories in which $\mathcal{C}$ is compact closed and furthermore in which the cups and caps occur as instrument elements up to a scalar. We refer to such generalised theories, which have in essence \textit{physical} cups and caps, as generalised theories with channel-state duality \cite{wilson2026supermapsgeneralisedtheories}.

\begin{definition}[Channel-State Duality]
A generalised theory has channel-state duality if its non-deterministic theory $\cat{C}$ is a compact closed category such that:
\begin{itemize}
\item The cup $\cup$ is an instrument element up to a scalar \[  \input{figs/cup.tikz} \ = \ \input{figs/cup_instrument.tikz}\]
\item The cap $\cap$ is an instrument element up to a scalar  \[ \input{figs/cap.tikz} \ = \ \input{figs/cap_instrument.tikz}\]
\end{itemize}
\end{definition}
Classical theory, quantum theory, and Boxworld are each generalised theories with channel-state duality. 

\subsection{Process functionals}
An assignment from parameterised instruments to probability distributions represents the possibility for the spacetime environment to be probed by a classical interface, or equivalently, the possibility for a spacetime environment to be applied in the context of an environment consisting of classical systems. In the following intuitive picture, and later formally, we will denote classical variables by red wires.
\begin{figure}[htbp]
\centering

\input{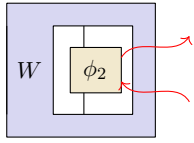}

\caption{A process functional can be applied in the context of any classical environment.}
\label{fig:classical-environment}
\end{figure}

In the simplest possible terms, the direct interpretation of a process functional is that it takes an instrument $\{M_x^y\}_x^y$ and returns a conditional probability distribution (stochastic channel), $ W_{\{ M_x^y \}_x^y }(y \vert x)$. We will denote $ W_{\{ M_x^y \}_x^y }(y \vert x)$ with the shorthand $ W_{M }(y \vert x)$ whenever the meaning is clear. 
In the spirit of Gleason's theorem for POVMs, we refer to the probability assignment as non-contextual if $M_a^b = N_{a'}^{b'} \implies W_{M}(b|a) = W_{N}(b'|a')$, so that we may write $W(M_a^b) := W_{M}(b|a)$.


In principle, we may define process functionals over general instrument sets. We will however choose to fix our attention to instrument sets which are closed under coarse-graining and convex combinations, where coarse-graining is defined by straightforward addition of sub-normalised processes. Noting that coarse-graining changes the outcome set of an instrument, it is more convenient to work with entire assignments $\mathcal{M}(-,=)$ from input/output sets such as $\mathbb{X},\mathbb{X}'$ to sets of instruments $\mathcal{M} \subseteq \{ M: A \otimes \mathbb{X} \rightarrow A' \otimes \mathbb{X}'   \}$. We will refer to such assignments, closed under convex combinations and coarse-graining, as simply \textit{sets of parameterised instruments}. In generalised theories, the set of all instruments $(A \Rightarrow_{inst} A')(-,=)$ defined by $(A \Rightarrow_{inst} A')(\mathbb{X},\mathbb{X}') :=  \{ M: A \otimes \mathbb{X} \rightarrow A' \otimes \mathbb{X}'   \}$ is always closed under convex combinations and coarse-graining. 

In the spirit of the Busch-Gleason theorem for POVMs \cite{Busch_2003}, we refer to a probability assignment as non-contextual if $M_a^b = N_{a'}^{b'} \implies W_{M}(b|a) = W_{N}(b'|a')$, so that we may write $W(M_a^b) := W_{M}(b|a)$. We can then refer to the following three desirable properties of causality, convexity, and coarse-graining, respectively, to define well-behavedness for such an assignment. 
\begin{definition}[Abstract Process Functional]
    A process functional on a set of parameterised instruments $\mathcal{M}(-,=)$ is a non-contextual assignment $W_{M}(y|x) := W(M_x^y) \in \mathbb{R}^{+}$ on the elements of instruments in $\mathcal{M}$ which is
    \begin{itemize}
\item Causal: $\sum_{y} W(M_x^y) = 1$,
\item Convex: $W(pM_{x_1}^{y} + (1-p)M_{x_2}^{y}) = pW(M_{x_1}^{y} ) + (1-p)W(M_{x_2}^{y} )$,
\item Compatible with coarse-graining: $W(M_a^{b_1} + M_a^{b_2}) = W(M_{a}^{b_1}) + W(M_{a}^{b_2})$.
\end{itemize}
\end{definition}
The latter condition on compatibility with coarse-graining is actually entailed by the first condition regarding normalisation; however, we include it for completeness to mirror the standard approach to motivating process matrices. To see this, note that for any instrument with components $\{ M^{y_i} \}$ we can coarse-grain two of its elements to construct a new instrument. For readability, we consider the coarse-graining over $y_1$ and $y_2$, which returns an instrument that for all $i>2$ has components identical to those of the original instrument, but in which $y_1$ and $y_2$ are replaced by a single element $y$ with associated component $M^{y} = M^{y_1} + M^{y_2}$. By normalisation, we then have:
\[   W(M^{y}) + \sum_{i>2} W(M^{y_i}) = 1 = W(M^{y_1}) + W(M^{y_2}) + \sum_{i>2} W(M^{y_i}) \] 
which entails that
\[ W(M^{y})  = W(M^{y_1}) + W(M^{y_2}) . \]

One can take compatibility with coarse-graining further, to the notion of compatibility with pre- and post-composition with pure functions. Let us define, for any function $g: Y \rightarrow Y'$, the induced instrument $M^{g^{-1}}$ by \[ (M^{g^{-1}})_x^{y'} := \sum_{k \in g^{-1}(y')}M_x^{k}  \] then we see that (using the derived compatibility with coarse-graining)
\begin{align}
 W_{M^{g^{-1}}}(y'|x) & = W((M^{g^{-1}})_x^{y'} )  \\ 
 & = W\bigg(\sum_{k \in g^{-1}(y')}M_x^{k}\bigg)  \\
 & = \sum_{k \in g^{-1}(y')} W(M_x^{k}) \\
 & = \sum_{k \in g^{-1}(y')} W_{M}(k|x) \\
 &   =  (g^{-1} \circ W_{M})(y'|x) .
\end{align}
For any function $f:X \rightarrow X'$ we can define $M_f$ by \[ (M_f)_{x}^{y} :=  M_{f(x)}^{y} .  \] We then see immediately that, in fact,
\[  W_{M_f}(y|x) =  W((M_f)_{x}^{y}) =  W(M_{f(x)}^{y})  =  W_{M}(y|f(x)) = (W_M \circ f)(y|x) . \]

\subsection{The problem with product instruments}

Of most interest to the study of indefinite causality and quantum causal structure is the multipartite setting in which a spacetime environment is thought of as being compatible with multiple local quantum laboratories. 
Note that it is easy to generalise process functionals to the product-input setting. We simply require that $W(M_x^{y}, N_{x'}^{y'})$ is convex-linear in each component, compatible with coarse-graining in each input, and normalised in the sense that for all $x,x'$ we have \[\sum_{y,y'} W(M_{x}^{y}, N_{x'}^{y'}) = 1. \] 
A family of functions on products $W(M_{x}^{y}, N_{x'}^{y'} )$ satisfying the above requirements with respect to coarse-graining and convexity component-wise has a representation in terms of a linear map $\rho_W$ via a Born rule, meaning that $W(M_{x}^{y}, N_{x'}^{y'} ) = Tr[(M_{x}^{y} \otimes N_{x'}^{y'} ) \rho_W ]$. This linear map, however, is not a fully positive operator, but instead a POPT tensor.
\begin{figure}[htbp]
\centering

 \[ \input{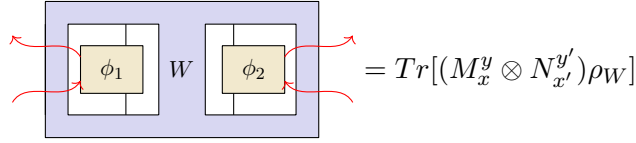} = Tr[(M_{x}^{y} \otimes N_{x'}^{y'} ) \rho_W ] \]

\caption{Any product process functional on quantum systems has a linear (non-positive) representation in terms of a POPT tensor.}
\label{fig:product-popt}
\end{figure}

Let us be precise about what exactly is recovered by product-input process functionals. 
\begin{theorem}
The following characterisations hold for product-input process functionals:
    \begin{itemize}
        \item Product-input process functionals on classical theory = classical stochastic processes
        \item Product-input process functionals on quantum theory = POPT tensors
        \item Product-input process functionals on $\mathbf{Box}$ = Boxworld process tensors. 
    \end{itemize}
\end{theorem}
\begin{proof}
The first statement follows by the Yoneda lemma for categorical supermaps, the second is a known result of the generalised Gleason theorem of \cite{barnum2005influencefreestatescompoundquantum}, and a proof of the last statement is given in the Appendix.
\end{proof}
Whilst the characterisation for higher-order maps here is correct for the classical case, in both the quantum case and for Boxworld, the resulting characterisations are non-physical. In general, not all POPT tensors are considered legitimate higher-order maps on quantum theory, and not all Boxworld process tensors are considered legitimate higher-order maps on Boxworld \cite{bavaresco2024indefinitecausalorderboxworld}. 

At first glance, the failure to capture the physicality of higher-order maps appears to take on a rather different flavour in the quantum and Boxworld cases. In the quantum case, the issue is a failure of positivity $\rho_W \geq 0$ for the linear map representation $\rho_W$. In the Boxworld case, it is not a failure of positivity but instead a failure of signalling properties, namely, a weakening of the no-signalling-without-system-exchange property, denoted wNSWSE \cite{bavaresco2024indefinitecausalorderboxworld}.
As the Yoneda lemma for categorical supermaps \cite{wilson2026supermapsgeneralisedtheories} suggests, however, this is really a failure in both cases of \textit{completeness}. In order to capture completeness within the framework of process functionals, we will need to turn our attention to process functionals over spaces which extend the set of product instruments, as suggested by the form of the tensor product for product spaces in the category of strong profunctors \cite{hefford_supermaps}.

\subsection{Multipartite process functionals}


To ensure full positivity of the process-matrix representation $\rho_W$ for a product process functional on quantum instruments, one further requires that it should produce legitimate probability distributions when applied to product instruments that are allowed to be entangled as mediated via an environment \cite{oreshkov}.
 In other words, one chooses only to study the process functionals such that furthermore $Tr[(M_{xz}^{y} \otimes N_{x'z'}^{y'} ) \rho_W \otimes \rho^{zz'} ] \geq 0$ for any choice of extended instruments $M,N$ and auxiliary state $\rho$. 
 
 Schematically (although the following schema is not yet formal), what is added as an additional assumption is well-behaviour of instruments which share a common past.
 \begin{figure}[htbp]
\centering

 \[  \input{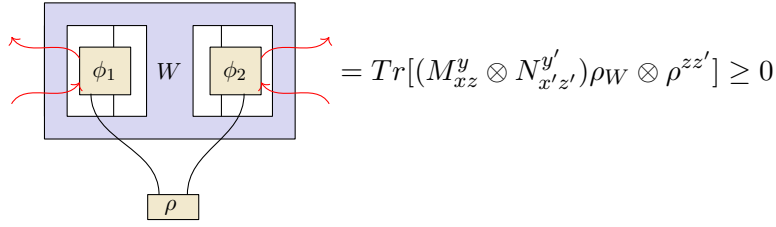} = Tr[(M_{xz}^{y} \otimes N_{x'z'}^{y'} ) \rho_W \otimes \rho^{zz'} ] \geq 0   \]

\caption{Process matrices are traditionally defined as the POPT tensors which satisfy a further positivity constraint \cite{oreshkov}.}
\label{fig:process-matrix-positivity}
\end{figure}

 Formally, note that this positivity condition is traditionally framed with respect to an already derived linear (but not yet positive) representation map $\rho_W$ \cite{oreshkov}. Since the existence of even this intermediate representation is not clearly present within other generalised theories, we are forced to take an alternative route to finding a more careful axiomatisation for process functionals, which encodes this compatibility with experiments involving shared resources, without invoking concepts such as positivity and internal linear representation. 
 
With this in mind, we are actually pushed towards an arguably cleaner axiomatisation even for multipartite process functionals on quantum theory. 
 \begin{definition}
 A ($2$-arity) multipartite process functional is a process functional over the following instrument set:
 \[ 
    \mathcal{M}^2(\mathbb{X}, \mathbb{X}') \ :=  \ \left\{\input{figs/spr_product_new.tikz}\right\}, \] 
    more generally, a multipartite process functional (of $n$-arity) is a process functional over the following instrument set 
     \[ 
    \mathcal{M}^n(\mathbb{X}, \mathbb{X}') \ :=  \ \left\{\input{figs/spr_product_narity.tikz}\right\}. \] 
 \end{definition}
 We can, in fact, take this approach further. Traditionally process matrices with a definite causal order (combs \cite{chiribella_circuits, chiribella_networks, process_tensor}) can be taken to be process matrices which satisfy further constraints. We can, however, axiomatise definite-order process functionals directly, without any reference to internal representation. 
 \begin{definition}
 A definite-order multipartite process functional is a process functional over the following instrument set:
 \[ 
    \mathcal{M}^2_{A \leq B}(\mathbb{X}, \mathbb{X}') \ :=  \ \left\{\input{figs/spr_definite.tikz}\right\}. \] 
 More generally, a total-order multipartite process functional (of $n$-arity) is a process functional over the following instrument set 
     \[ 
    \mathcal{M}^n_{A_1 \leq  \dots \leq A_n}(\mathbb{X}, \mathbb{X}') \ :=  \ \left\{\input{figs/spr_definite_n.tikz}\right\}. \] 
 \end{definition} 
 One can naturally build on these examples to define a myriad of other examples of instrument sets which encode causal constraints for any process functional that might act upon them. 
\begin{theorem}
The following characterisations hold:
\begin{itemize}
\item Multipartite process functionals on quantum theory are equivalent to multipartite process matrices.
\item Definite-order multipartite process functionals on quantum theory are equivalent to the multipartite process matrices with definite causal order.
\end{itemize}
\end{theorem}
\begin{proof}
A direct proof can be found in \cite{Hoffreumon2023}, furthermore, it is a direct consequence of the process-matrix representation theorem for generalised theories which will be given in Section~4. 
\end{proof}
Furthermore, we will find that we can characterise the generalised theories for which multipartite process functionals support a process-matrix-like representation. This more general theorem will recover the above when applied to quantum theory. 

\section{Process correlations from commutativity}

We now proceed to develop characterisation theorems for process functionals. We begin with the most general structural correspondence, which holds in any generalised theory and relates process functionals directly to categorical supermaps \cite{wilson2026supermapsgeneralisedtheories}. We will find that the process-functional axioms are equivalent to local applicability \cite{wilson_locality} (commutativity) over classical variables. 

To begin, let us note a purely compositional interpretation of parameterised instrument sets. As viewed from the outside, there are a few simple compositional axioms that one might expect of a class of experiments available to a lab: \[   \input{figs/sub_prof_1.tikz}    \]
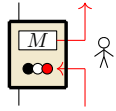
\begin{itemize}
\item First, given any experiment that can be performed in the lab, we expect that one could construct a new experiment by pre- or post-composing with classical processes, representing the possibility of relabelling and coarse-graining settings and outcomes respectively: \[   \input{figs/sub_prof_2.tikz}    \]
\item Furthermore, one might expect that, given any legitimate experiment, one could construct a new legitimate experiment by simply allowing a further classical environment to pass through untouched, representing the possibility of passing additional information outside the inner workings of the experiment: \[   \input{figs/sub_prof_3.tikz}    \]
\end{itemize} 
Together these possibilities suggest that any reasonable set of experiments, or interventions, should be closed under all classical combs: \[   \input{figs/sub_prof_4.tikz}    \]

A family of sets equipped with an action by classical combs in this way can be modelled abstractly as a \textit{strong profunctor} \cite{benabou, tambara, pastro_street} $\mathcal{M}: \opcat{\mathbf{Stoch}} \times \mathbf{Stoch} \rightarrow \mathbf{Set}$ over classical variables (classicality is expressed by having variables live in the category $\mathbf{Stoch}$). This connects the study of instruments to the study of categorical supermaps \cite{hefford_supermaps}, since it has been established that all higher-order quantum types can be thought of as strong profunctors $\mathcal{P}: \opcat{\mathbf{CPTP}} \times \mathbf{CPTP} \rightarrow \mathbf{Set}$ over quantum variables (where quantum-ness is expressed by having variables live in the category $\mathbf{CPTP}$ of CPTP maps).
One can see quite directly that any sub-strong profunctor $\mathcal{M}(\mathbb{X} , \mathbb{X}')  \subseteq \{ M : A \otimes  \mathbb{X}  \rightarrow A' \otimes \mathbb{X}'   \}$ (any set of instruments closed under action by stochastic combs) is a parameterised instrument set in the sense of the previous section.

A purely compositional notion of compatibility with classical experimental choices and outcomes is then commutation with such pre- and post-processing combs. 



\begin{definition}[Commutation Functional]
A \textit{commutation functional} on a parameterised instrument set $\mathcal{M}$ is a strong natural transformation of type \[   W_{\mathbb{X}\mathbb{X}'}: \mathcal{M}(\mathbb{X} , \mathbb{X}') \rightarrow \mathbf{Stoch}(\mathbb{X} , \mathbb{X}') .   \] 
\end{definition}

Let us now be more concrete and unpack this definition. A commutation functional is a family of functions from instruments to conditional probability distributions which commutes with classical combs. Even more concretely, a commutation functional is a family of functions $W_{\mathbb{X} \mathbb{X}'} :  \mathcal{M}(\mathbb{X} , \mathbb{X}') \rightarrow  \stoch(\mathbb{X} , \mathbb{X}')$, depicted graphically as \[ \input{figs/w_1.tikz}  \] satisfying \[  \input{figs/w_2.tikz}  \ = \   \input{figs/w_3.tikz}.   \] The only graphical difference from standard categorical supermaps \cite{wilson2026supermapsgeneralisedtheories} is that we have required commutation only with red wires, that is, only with respect to classical-variable environments. 

Just as closure under action by combs is equivalent to closure under coarse-graining and convex combinations, commutation with combs on classical variables is precisely the requirement of having a non-contextual assignment which is compatible with coarse-graining and convex combinations.

\begin{theorem}
On any generalised theory $\cat{C}$:
\[ \textrm{Process functionals = Commutation functionals }. \]
\end{theorem}
\begin{proof}
Given in the Appendix.
\end{proof}

The theorem gives us a conceptual and practical takeaway. Conceptually, the traditional axioms of process functionals can be recast into a purely compositional principle, namely commutativity with classical combs. Practically, this commutation framing allows us to reason with process functionals diagrammatically, using function boxes for string diagrams.

\section{Process matrix representation in generalised theories}
In this section we ask a fundamental question regarding process functionals on generalised theories: do they support a process-matrix representation? In other words, for any process functional $W$, does there exist some notion of representative state $\rho_{W}$? Recall that in the quantum setting, such a state is required to satisfy two conditions:
\begin{itemize}
\item $\rho_W \geq 0$,
\item  $Tr[\rho_{W} M_{a}^{b}  ] = W(M_{a}^{b})$.
\end{itemize}
The positivity requirement is simply the requirement that the linear operator $\rho_W$ be a legitimate state of quantum theory. As a result, we can quite directly state what it means to give a process-matrix representation in any theory. 
\begin{definition}[Process Matrix on a generalised theory]
For any generalised theory with channel-state duality and any instrument sub-profunctor $\mathcal{M}$, an abstract process functional $W$ is said to have a process-matrix representation if there exists a state $\rho_{W}$ such that, \[ \input{figs/box_tensor_6.tikz} \ = \  \input{figs/pm_rep_1.tikz}  \] 
\end{definition}
The traditional proof of the existence of a process-matrix representation for process functionals on quantum theory can be thought of as making judicious use of the Busch-Gleason theorem, which states that any frame function over POVM elements is a quantum state. 
Here, a frame function is taken to be a function $f(E^y)$ such that $\sum_{y}f(E^y) = 1$. It is quite direct to see that such functions on POVMs are convex linear and compatible with coarse-graining, so that the Busch-Gleason theorem really can be understood as stating that process functionals on destructive measurements, that is, measurements of the form $E: A \rightarrow I \otimes \mathbb{X}$ are always represented by quantum states by $f(E^y) = Tr[\rho_W E^y]$. 

Noting that the notion of frame function and the notion of internal representative state exist for any generalised theory, we now lift the Busch-Gleason theorem for quantum theory to the level of an axiom on generalised theories. We will see that this axiom is both necessary and sufficient to prove a process-matrix representation. 
\begin{definition}[Busch-Gleason Axiom for generalised theories]
We say that a generalised theory satisfies the Busch-Gleason axiom if for every frame function there exists a representative state $\rho_{f}$ such that $f(E^y) = E^y \circ \rho_f$.
\end{definition}
Note that, by the equivalence between commutation functionals and process functionals (which in turn specialise to frame functions), the Busch-Gleason axiom can be reinterpreted as stating that for any family of functions on destructive measurements which commutes with classical post-processing in the following sense \[    \tikzfigscale{0.7}{figs/busch_gleason_1} \ = \ \tikzfigscale{0.7}{figs/busch_gleason_2} \] there exists a representative state $\rho_{f}$ satisfying \[     \tikzfigscale{0.7}{figs/busch_gleason_0} \ = \ \tikzfigscale{0.7}{figs/busch_gleason_3}. \]

The Busch-Gleason axiom as we have stated requires the existence of a representative state, but not necessarily uniqueness of that state. The requirement that a representative state be unique is equivalent to the requirement that states be topologically distinguishable. 
\begin{definition}
A pair of states $\rho, \sigma : I \rightarrow A$ are tomographically distinguishable if there exists some destructive measurement $E : A \rightarrow {\color{red}\mathbb{X}}$ such that $E \circ \rho \neq E \circ \sigma$. A generalised theory is termed tomographically distinguishable if every pair of states is tomographically distinguishable. 
\end{definition}
When a generalised theory is tomographically distinguishable, this precisely means that any pair of states $\rho, \sigma$ represent distinct frame functions (as witnessed by evaluating on the distinguishing measurement $E$). Therefore, tomographic distnguishability is equivalent to uniqueness of the Busch-Gleason representation.


We will now proceed to characterise process functionals in a large variety of cases.
\begin{lemma}[Process Matrix Representation for Causal Process Functionals]
Consider any generalised theory $\cat{C}$ satisfying,
\begin{itemize}
\item Channel-state duality,
\item The Busch-Gleason axiom.
\end{itemize}
Consider any parameterised instrument set $\mathcal{M}$ in $\cat{C}$ which contains the non-signalling scenario  \[   \input{figs/in_instrument_1.tikz} \ \in \ \mathcal{M} . \]
Then every abstract process functional $W$ on $\mathcal{M}$ has a process-matrix representation, furthermore, if the generalised theory is tomographically distinguishable, then this representation is unique. 
\end{lemma}
\begin{proof}
Given in the Appendix. 
\end{proof}
We now specialise this representation to multipartite process functionals. 
\begin{theorem}
Consider any generalised theory $\cat{C}$ with channel-state duality. The following are equivalent:
\begin{itemize}
\item Busch-Gleason axiom,
\item For every one-input process functional $W$ there exists a state $\rho_W$ such that: \[ \input{figs/box_tensor_6.tikz} \ = \  \input{figs/pm_rep_1.tikz}, \] 
\item For every definite-order multipartite process functional $W$ there exists a state $\rho_W$ such that: \[   \input{figs/process_thm_def.tikz}  \ = \   \input{figs/process_thm_def2.tikz}  .  \] 
\item For every multipartite process functional $W$ there exists a state $\rho_W$ such that: \[   \input{figs/process_thm_1.tikz}  \ = \   \input{figs/process_thm_2.tikz}  .   \]
\end{itemize}
Furthermore, when the generalised theory is tomographically distinguishable, each representation $\rho_W$ is unique. 
\end{theorem}
\begin{proof}
This amounts to simply checking that any $n$-arity instrument set contains the no-signalling scenario, a full proof of this fact is given in the Appendix. 
\end{proof}
Consequently, we have arrived at a clean characterisation of the theories which support a process-matrix representation theorem. 

\section{When are supermaps and process matrices equivalent?}

An alternative viewpoint to the process functional approach, commonly associated with the term \emph{quantum superchannel}, ignores direct reference to probabilistic structure and instead concerns the generalisation of completeness to transformations between quantum operations. At a higher level, this completeness is the requirement that a spacetime environment can be applied to part of any bipartite process, or equivalently, in the context of an environment consisting of physical systems. 

\begin{figure}[htbp]
\centering

\[  \input{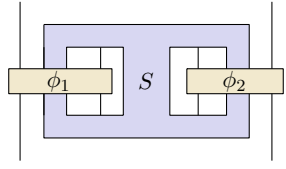} \]

\caption{A supermap is a complete operation on deterministic operations, that is, it can be applied in the context of any physical environment.}
\label{fig:supermap-environment}
\end{figure}

We give a precise linear algebraic definition for quantum superchannels in the appendix.
The linear algebraic approach can be translated into the language of categorical quantum mechanics, and in turn be generalised to a broader class of theories, so long as those theories exhibit a channel-state duality. 
\begin{definition}
    A CJ-supermap of type \[ (A \Rightarrow A') \rightarrow (B \Rightarrow B') \] on a generalised theory with channel-state duality is a process of type $S: B \otimes A' \rightarrow B' \otimes A$ in $\cat{C}$ such that for every $\phi : A \otimes X \rightarrow A' \otimes X'$ in $\cat{C}^{\mathbf{d}}$ the following map is also in $\cat{C}^{\mathbf{d}} $. \[ \input{figs/cj_supermap_1.tikz} \]
\end{definition}
This definition generalises to the multi-input case \cite{wilson2026supermapsgeneralisedtheories}. A CJ-supermap of type \[ \bigtimes_{i}(A_i \Rightarrow A_i') \rightarrow (B \Rightarrow B') \] on a generalised theory is a process of type $S: B \otimes (A_1' \otimes \dots \otimes A_n') \rightarrow B' \otimes (A_1 \otimes \dots \otimes A_n)$ in $\cat{C}$ such that for every collection $\phi_i : A_i \otimes X_i \rightarrow A_i' \otimes X_i'$ of processes in $\cat{C}^{\mathbf{d}} $ the following process is also in $\cat{C}^{\mathbf{d}} $. \[ \input{figs/cj_supermap_2.tikz}. \] 
In the case of quantum theory, this is verbatim the standard definition of quantum supermaps, and similarly for classical supermaps. In a recent work \cite{wilson2026supermapsgeneralisedtheories} it was proven that such maps on Boxworld are characterised by a physically well-motivated weakening of the NSWSE condition for Boxworld tensors used for the study of indefinite causal order in Boxworld \cite{bavaresco2024indefinitecausalorderboxworld}. This weakening was denoted cwNSWSE, shorthand for tensors that satisfy wNSWSE and continue to do so under tensor extension, in other words tensors that are \textit{completely} wNSWSE.

Supermaps can even be generalised to theories which do not exhibit a channel-state duality. A categorical supermap on a general symmetric monoidal category can be defined as a locally applicable transformation, which we from now on refer to simply as a categorical supermap. 
\begin{definition}[Categorical Supermaps]
A one-input categorical supermap of type \[S: (A \Rightarrow A') \rightarrow (B \Rightarrow B')\] is a family of functions \[   S_{X,X'} :   \cat{C}^{\mathbf{d}}(A \otimes X, A' \otimes X') \rightarrow \cat{C}^{\mathbf{d}}(B \otimes X , B' \otimes X')  \] satisfying  \[   \input{figs/lot_1_outside.tikz} \ = \   \input{figs/lot_1_inside.tikz} . \]
\end{definition}
This abstract approach also has a multi-party generalisation \cite{wilson2026supermapsgeneralisedtheories}: a multi-input categorical supermap of type \[    S:   \bigtimes_i (A_i \Rightarrow A_i') \rightarrow (B \Rightarrow B')  \] is a family of functions satisfying  \[   \input{figs/lot_multi_outside.tikz} \ = \ \input{figs/lot_multi_inside.tikz}.    \]
The categorical supermap definition can in fact be applied to any strong profunctor, that is, any space closed under action by combs. The CJ-supermap definition can be applied to any sub-strong profunctor of the canonical partial-hom profunctor $\mathcal{C}(A \otimes - , A'\otimes = )$. 

\begin{lemma}[Supermaps define probabilistic correlations]
For any generalised theory $\cat{C}$, every multi-input categorical supermap on $\cat{C}$ with trivial output type $I \Rightarrow I$ defines a product process functional.
\end{lemma}
\begin{proof}
Follows simply by keeping the evaluation of the categorical supermap on classical wires \[  W(M_1,M_2) \  := \   \input{figs/sup_proc_1.tikz}  . \]
\end{proof}
\begin{lemma}[Supermaps are process functionals]
For any generalised theory $\cat{C}$ with channel-state duality, every multi-input categorical supermap on $\cat{C}$ with trivial output type $I \Rightarrow I$ defines a multipartite process functional.
\end{lemma}
\begin{proof}
Follows from the Yoneda Lemma for categorical supermaps. We define the action on general multi-partite instrument shapes by  \[ W \left( \input{figs/sup_proc_a.tikz}  \right) \  := \   \input{figs/sup_proc_b.tikz} \] however, we must check that this assignment is well-formed. To do so, let us suppose that \[  \input{figs/sup_proc_2.tikz} \  = \   \input{figs/sup_proc_3.tikz} \] then we have that \[  W \left( \input{figs/sup_proc_2.tikz}  \right) \ = \   \input{figs/sup_proc_4.tikz} \  = \   \input{figs/sup_proc_5.tikz}  \]
\[   = \    \input{figs/sup_proc_6.tikz} \  = \   \input{figs/sup_proc_7.tikz}  \ = \ W \left( \input{figs/sup_proc_3.tikz}  \right). \]
\end{proof}

To find a characterisation of the generalised theories on which supermaps are equivalent to process matrices, we will need a final basic principle, tomographic closure. 
\begin{definition}
A process $M: A \rightarrow A'$ is \textit{tomographically deterministic} when the concatenation  \[   \input{figs/stat_closed_1.tikz},   \] of $M: A \rightarrow A'$ with arbitrary deterministic preparations $\rho$ and measurements $\sigma$ returns a deterministic classical channel.
\end{definition}
\begin{definition}
A generalised theory is tomographically closed when the following statements are equivalent:
\begin{itemize}
\item $M: A \rightarrow A'$ is deterministic
\item $M: A \rightarrow A'$ is tomographically deterministic
\end{itemize}
\end{definition}
For a theory to be tomographically closed means that any physical process which produces normalised statistics is a deterministic process. Note furthermore that classical theory, quantum theory, and Boxworld are easily confirmed to be tomographically closed.

Note that we will use the word supermaps on a generalised theory unambiguously. By the Yoneda lemma for categorical supermaps there really is no difference between strong natural transformations and CJ-supermaps. 
\begin{theorem}
Consider a generalised theory $\cat{C}$ with channel-state duality:
\[ \textrm{ $\cat{C}$ is tomographically closed}  \]
\[  \iff  \]
\[  \{ \textrm{Multipartite Process Matrices} \}  \cong  \{ \textrm{Supermaps of type $S: \bigtimes_i (A_i \Rightarrow A_i') \rightarrow (I \Rightarrow I)$} \}   \]
\end{theorem}
\begin{proof}
Given in the Appendix.
\end{proof}
We can furthermore combine the process-matrix representation in the presence of the Busch-Gleason axiom to conclude that for any generalised theory $\cat{C}$ with channel-state duality, satisfaction of the Busch-Gleason axiom together with tomographic conditions (namely closure and distinguishability) are equivalent to requiring that multipartite process functionals are equivalent to supermaps. 
\begin{theorem}
Consider a generalised theory $\cat{C}$ with channel-state duality:
\[ \textrm{Busch-Gleason + Tomographic Distinguishability + Tomographic Closure}  \]
\[  \iff  \]
\[  \{ \textrm{Multipartite Process Functionals} \}  \cong  \{ \textrm{Supermaps of type $S: \bigtimes_i (A_i \Rightarrow A_i') \rightarrow (I \Rightarrow I)$} \}   \]
\end{theorem}
\begin{proof}
By the previous theorem, Tomographic Closure implies
\[
\{ \textrm{Multipartite Process Matrices} \} \cong \{ \textrm{Supermaps of type $S: \bigtimes_i (A_i \Rightarrow A_i') \rightarrow (I \Rightarrow I)$} \}.
\]
By the process-matrix representation theorem, the Busch-Gleason axiom implies
\[
\{ \textrm{Multipartite Process Functionals} \} \cong \{ \textrm{Multipartite Process Matrices} \}.
\]
Conversely, if
\[
\{ \textrm{Multipartite Process Functionals} \} \cong \{ \textrm{Supermaps of type $S: \bigtimes_i (A_i \Rightarrow A_i') \rightarrow (I \Rightarrow I)$} \},
\]
then every multipartite process functional has a process-matrix representation, meaning that the Busch-Gleason axiom must hold, and every such process matrix is a supermap, meaning that Tomographic Closure holds. This completes the proof. 
\end{proof}

This in turn leads us to a concrete characterisation of multipartite process functionals in classical theory, quantum theory, and Boxworld. 
\begin{theorem}
The following characterisations of multipartite process functionals hold:
    \begin{itemize}
        \item Multipartite process functionals on classical theory = classical supermaps / stochastic processes,
        \item Multipartite process functionals on quantum theory = quantum supermaps / process matrices,
        \item Multipartite process functionals on $\mathbf{Box}$ = cwNSWSE Boxworld tensors. 
    \end{itemize}
\end{theorem}
\begin{proof}
Note that classical theory, quantum theory, and Boxworld each satisfy the Busch-Gleason axiom and Tomographic Closure. As a result, multipartite process functionals on classical theory, quantum theory, and Boxworld are equivalent to the corresponding categorical supermaps. The Yoneda lemma for categorical supermaps and the concrete characterisation of CJ-supermaps from \cite{wilson2026supermapsgeneralisedtheories} entail that supermaps on each of these three theories are indeed classical supermaps / classical processes, quantum supermaps / process matrices, and the cwNSWSE Boxworld tensors, respectively. This concludes the proof. 
\end{proof}

\section{Discussion}
Previous attempts to give a general approach to defining higher-order processes on finite-dimensional theories beyond quantum mechanics have focussed on compositional structure, effectively taking the supermap perspective. In this paper we considered a generalisation of the process-matrix perspective, in which the primary focus is shifted to the extraction of probability distributions from spatiotemporally constrained instrument sets interpreted as modelling local laboratories. 
We worked in the setting of generalised theories, a minimal class of theories supporting both composition and probability, and established some foundational theorems regarding their structure. 

First, we identified process functionals on generalised theories and proved that they can be reframed purely compositionally as naturality (commutation). Because of this, every categorical supermap on a generalised theory is a multipartite process functional. 
Next, we asked when every multipartite process functional on a theory has a process-matrix-like representation, and we found that, following the traditional derivation of the process-matrix representation, such theories are characterised by the simpler Busch-Gleason axiom. 
Finally, we brought together the process-matrix representation theorem and the connection between categorical supermaps and process functionals, characterising those theories on which multipartite process functionals are equivalent to supermaps as those satisfying both the Busch-Gleason axiom \textit{and} Tomographic Closure. 

As a consequence, the results of this paper establish basic conditions on generalised theories, namely the Busch-Gleason axiom and Tomographic Closure, under which there is no ambiguity whatsoever regarding the right notion of higher-order process. Theories included under this umbrella include classical theory, quantum theory, the theory of classical boxes (a Boxworld-like theory). 
Across a broad spectrum of general theories, then, this yields a clean framework for studying the fundamental limits of causality beyond finite-dimensional quantum mechanics.

 \section*{Acknowledgments}

 	This work has been partially funded by the French National Research Agency (ANR) within the framework of ``Plan France 2030'', under the research projects EPIQ ANR-22-PETQ-0007 and HQI-R\&D ANR-22-PNCQ-0002. 
The authors are grateful to Jessica Bavaresco, Cyril Branciard, Kuntal Sengupta, Alessandro Bisio, Luca Apadula, Zixuan Liu, and Anna Jencova for useful discussions.

\bibliographystyle{plainnat}
\bibliography{bibliography}

\appendix

\section{Commutation functionals and process functionals}

In this section, we will prove that commutation functionals (strong natural transformations over classical variables) are equivalent to the traditional formulation of process functionals. 
\begin{theorem}
On any generalised theory $\cat{C}$:
\[ \textrm{Process functionals = Commutation functionals }. \]
\end{theorem}
\begin{proof}

Let us first prove that process functionals define strong natural transformations; the remaining parts of the proof are either trivial or already present in the literature. Given an abstract process functional, we can prove commutativity on the bottom wire as follows. We see that for all $x,y$,
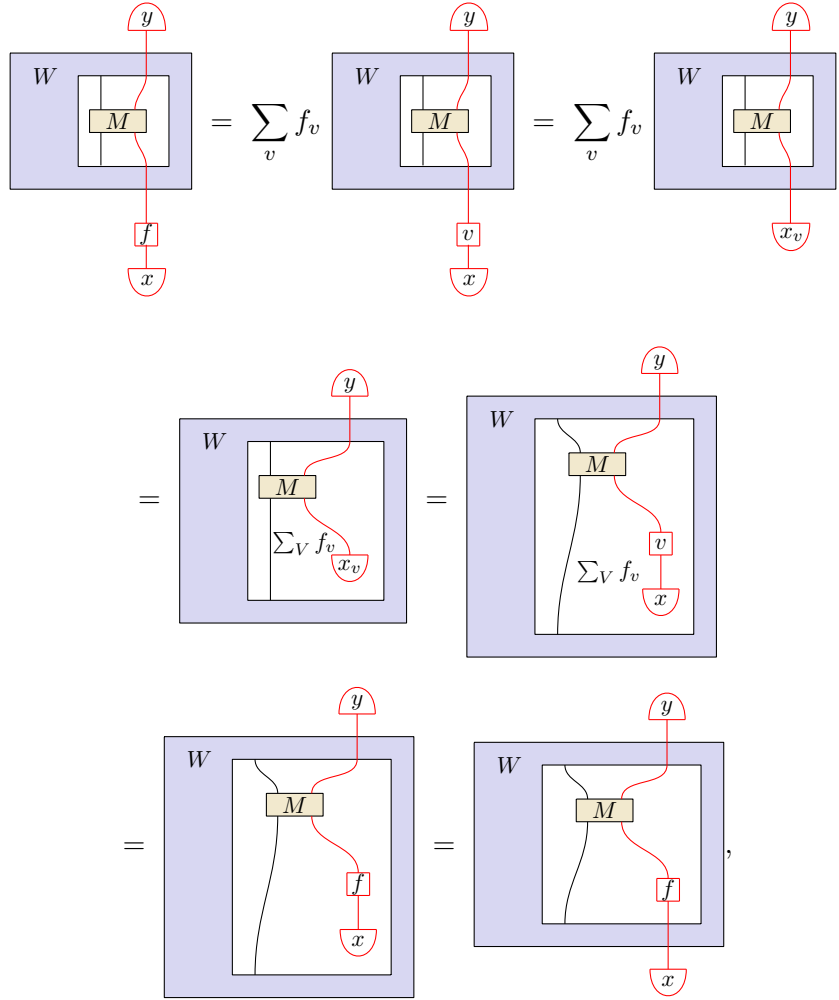
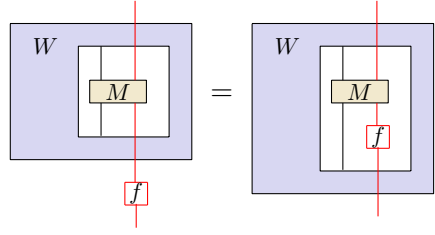
\[   \input{figs/w_com_1.tikz} \ = \ \sum_{v} f_v  \    \input{figs/w_com_2.tikz}  \ = \  \sum_{v} f_v  \   \input{figs/w_com_3.tikz}   \] 
\[ \ = \   \input{figs/w_com_5.tikz}  \ = \   \input{figs/w_com_6.tikz}   \] 
\[  \ = \   \input{figs/w_com_7.tikz} \ = \   \input{figs/w_com_8.tikz} ,  \] 
and so
\[   \input{figs/w_com_f1.tikz} \ = \   \input{figs/w_com_f2.tikz} .  \] 
Using commutation with any pure function on the top wire, we can prove commutativity with an arbitrary classical channel. Indeed, note that any classical channel is a convex combination of pure channels, and that any convex combination of pure channels can be re-written as a controlled channel with a probability distribution on the control wire, giving that 
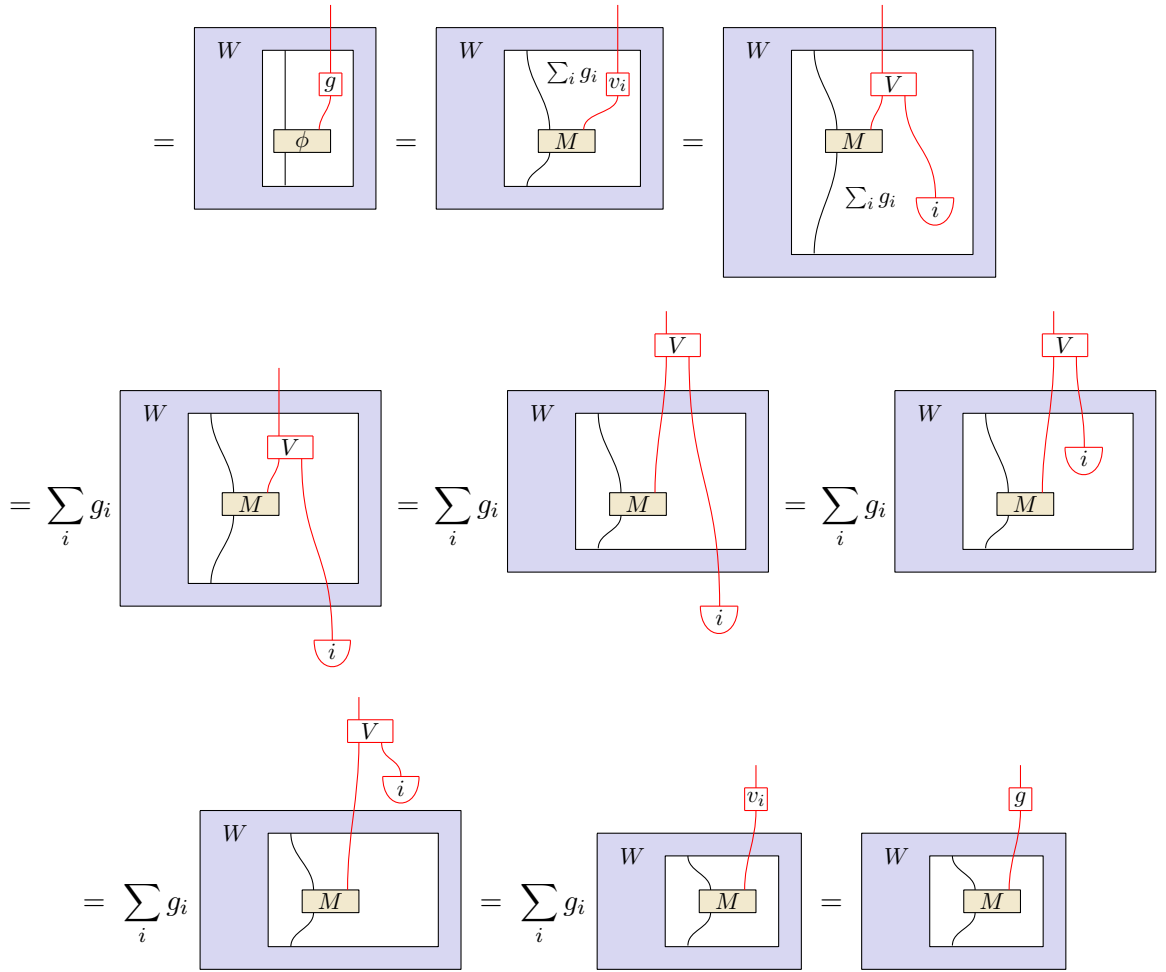
\[ \ = \   \input{figs/w_top_1.tikz}  \ = \   \input{figs/w_top_2.tikz}   \ = \   \input{figs/w_top_4.tikz}   \] 
\[ \ = \  \sum_{i} g_i \  \input{figs/w_top_5.tikz}  \ = \   \sum_{i} g_i \  \input{figs/w_top_6.tikz}   \ = \   \sum_{i} g_i  \ \input{figs/w_top_7.tikz}   \] 
\[ \ = \   \sum_{i} g_i  \ \input{figs/w_top_8.tikz}  \ = \   \sum_{i} g_i  \  \input{figs/w_top_8b.tikz}  \ = \   \input{figs/w_top_9.tikz}    \] 
Finally, we note that any natural transformation on an instrument is automatically strong, since \[   \input{figs/pm_thm_20.tikz} \ = \   \input{figs/pm_thm_21.tikz}  \ = \   \input{figs/pm_thm_22.tikz},   \] from which we can conclude that \[   \input{figs/pm_thm_23.tikz} \ = \   \input{figs/pm_thm_24.tikz}.   \] 

To complete the proof, note that by the proof method of \cite{wilson2026supermapsgeneralisedtheories} for the non-contextual assignment lemma (which, when closely examined, can be seen to leverage only commutativity over classical wires), every strong natural transformation is non-contextual on probability distributions, where we define $M_{p}^{q}$ with $p,q$ probability distributions over the basic indices $p = \{  (p_x,x)  \}$ and $q = \{ (q_y, y) \}$ by
 \[M_{p}^{q} = \sum_{x,y}  p_xq_y M_{x}^{y}. \]
A non-contextual-on-distributions assignment is then naturally a function $W_{M}(q|p) := W(M_{p}^q)$ which is again normalised in the sense that $\sum_{y} W(M_{p}^y) = 1$. Finally, any non-contextual assignment of this kind is automatically non-contextual on point distributions, and is furthermore convex, since \[  \sum_{x}  p_x W_{M}(y|x) = W_{M}(y|\sum_{x}  p_x x)  = W(M_{\sum_{x}  p_x x}^{y}) = W(\sum_{x}  p_x M_{x}^{y}) . \] Note that normalisation follows by definition, and compatibility with coarse-graining follows for free (via commutation with pure functions).

\end{proof}

\section{Concrete representation lemmas}
In this section we establish and general Lemma which will be used for a variety of purposes, demonstrating that product-input process functionals on $\mathbf{Box}$ are equivalent to Boxworld tensors, demonstrating that $\mathbf{Box}$ satisfies the Busch-Gleason axiom, and giving a process matrix representation theorem for process functions of theories supporting the Busch-Gleason axiom.
\begin{lemma}[Process Matrix Representation for Process Functionals]
In any generalised theory satisfying
\begin{itemize}
\item Channel-state duality,
\item The Busch-Gleason axiom,
\end{itemize}
for any pair of instrument sub-profunctors $\mathcal{M}_i$ which contains all non-signalling scenarios  \[   \input{figs/in_instrument_1.tikz} \ \in \ \mathcal{M}_i , \]
and for any product process functional $W$ on $\mathcal{M}_1,\mathcal{M}_2$ one can construct a corresponding product frame functional, that is, a function on pairs of destructive measurements satisfying  \begin{align*}    
 \input{figs/new_lot_3p.tikz} \ = \  \input{figs/new_lot_6p.tikz}. 
\end{align*} 
\end{lemma}
\begin{proof}
From any process functional on an instrument set which includes the above non-signalling scenario, we can an assignment from products of destructive measurements to probability distributions which commutes with classical post-processing via \[   \input{figs/new_lot_0p.tikz} \ := \ \frac{\input{figs/new_lot_1p.tikz}}{\input{figs/new_lot_2p.tikz}}.   \] This assignment indeed commutes with post-composition since
 \begin{align*}    
 \input{figs/new_lot_3p.tikz} \ &= \ \frac{\input{figs/new_lot_4p.tikz}}{\input{figs/new_lot_5ap.tikz}} \\
&  = \  \frac{\input{figs/new_lot_4bp.tikz}}{\input{figs/new_lot_5zp.tikz}}  \\
& = \  \frac{\input{figs/new_lot_4bp.tikz}}{\input{figs/new_lot_5cp.tikz}} \\
& = \  \input{figs/new_lot_6p.tikz}. 
\end{align*} 
\end{proof}
This lemma has then three consequences. 

\begin{lemma}[Process Matrix Representation for Process Functionals]
In any generalised theory satisfying
\begin{itemize}
\item Channel-state duality
\item The Busch-Gleason axiom
\end{itemize}
Consider any instrument sub-profunctor $\mathcal{M}$ which contains the non-signalling scenario  \[   \input{figs/in_instrument_1.tikz} \ \in \ \mathcal{M}  \]
Every abstract process functional $W$ on $\mathcal{M}$ has a process-matrix representation, furthermore, if the generalised theory is tomographically distinguishable, the representation is unique. 
\end{lemma}
\begin{proof}
We consider the one-input specialisation. In that case, the Busch-Gleason axiom holds, and so we have that there exists a representative state $\rho_{CuW}$ such that  \[   \input{figs/new_lot_7.tikz} \ = \frac{  \input{figs/new_lot_1.tikz}}{ \input{figs/new_lot_x.tikz}}.  \]  noting that  \[   \input{figs/recover_denom_2.tikz} \ = \ \input{figs/recover_denom_3.tikz}.  \] and so by non-contextual assignment  \[   \input{figs/recover_denom_1.tikz} \ = \ \input{figs/recover_denom_4.tikz} \ = \ 1 , \] 
 we now test the state on a general instrument element and find that (where the above calculation shows that the denominator will be $1$)
 \[   \input{figs/recover_pm_1.tikz} \ = \ \input{figs/recover_pm_2.tikz} \ = \ \input{figs/recover_pm_3.tikz} \ = \   \input{figs/recover_pm_4.tikz}   \]  
 and so using the non-contextuality of assignments a final time,
 \[  = \  \input{figs/recover_pm_5.tikz}  .\] 
 To prove uniqueness, consider any pair of process matrices $\rho_{CuW}, \sigma_{CuW}$ that represent the same process functional, then we have that for every POVM $E$ then  \[   \input{figs/new_lot_7.tikz} \ = \   \input{figs/unique_2.tikz}  \ = \   \input{figs/unique_3.tikz}  \ = \   \input{figs/unique_4.tikz},  \] and so indeed by tomographic distinguishability $\rho_{CuW}, \sigma_{CuW}$. 
 \end{proof}

\begin{lemma}
The theory $\mathbf{Box}$ satisfies the Busch-Gleason axiom. 
\end{lemma}
\begin{proof}
Consider the Busch-Gleason axiom on $\mathbf{Box}$, it types as requiring any family of functions satisfying  \[    \tikzfigscale{0.7}{figs/busch_gleason_1p} \ = \ \tikzfigscale{0.7}{figs/busch_gleason_2p} \] is represented by evaluation on a state $\rho_{f}: 1 \rightarrow (A \Rightarrow A')$. Now note that unpacking up to isomorphism in classical theory, we can write  \[    \tikzfigscale{0.7}{figs/unpack_1} \ \cong \ \tikzfigscale{0.7}{figs/unpack_2}. \]  
Since this indeed contains the non-signalling scenario and classical theory satisfies the Busch-Gleason axiom respectively, we have that there exists a state $\rho$ satisfying  \[    \tikzfigscale{0.7}{figs/unpack_3} \ \cong \ \tikzfigscale{0.7}{figs/unpack_4}. \] Incidentally, this imposes that $\rho_{CuW}$ must be a channel (being normalised by all deterministic combs), and so, must be a state in $\mathbf{Box}$. The multipartite case holds by identical reasoning, identifying $E$ with classical supermaps, forcing the representative state $\rho{CuW}$ to be a non-signalling classical, and so, a multipartite state in $\mathbf{Box}$. 
\end{proof}

\begin{lemma}
The following characterisations hold for product-input process functionals:
    \begin{itemize}
        \item Product-input process functionals on classical theory = classical stochastic processes
        \item Product-input process functionals on quantum theory = POPT tensors
        \item Product-input process functionals on Boxworld = Boxworld process tensors. 
    \end{itemize}
\end{lemma}
\begin{proof}
For a fixed choice of $E_1$ note that  \[   \input{figs/new_lot_1p.tikz}  \] is a frame function over $E_2$ valued in distributions rather than scalars, meaning it take the form: \[   \tikzfigscale{0.7}{figs/form_1}. \] Now in the specific case of classical tensors, we can deduce this form of Busch-Gleason from the scalar version. Indeed, noting that by non-contextual assignment we have for the inclusion map $\leq : (A \Rightarrow A') \rightarrow A \otimes A'$ then \[   \tikzfigscale{0.7}{figs/form_1} \ = \   \tikzfigscale{0.7}{figs/form_2} . \]  using this same trick (which is not available for quantum systems, in which positivity constraints appear) on the left and right products then gives a representative state $\rho_{CuW}$ such that  \[   \input{figs/new_lot_7p.tikz} \ = \ \input{figs/new_lot_1p.tikz}.  \] 
Again by non-contextual assignment we have that
\begin{align*}
 & \input{figs/recover_denom_1p.tikz} \\
 = \ & \input{figs/recover_denom_4p.tikz} \\
  = \ & 1 . 
 \end{align*} 
and so we find that
 \begin{align*}
   \input{figs/recover_pm_1p.tikz} \ = & \ \input{figs/recover_pm_2p.tikz} \\
    =& \ \input{figs/recover_pm_3p.tikz} \\
     =&  \   \input{figs/recover_pm_4p.tikz} \\
      =&  \  \input{figs/recover_pm_5p.tikz}
\end{align*}

\end{proof}


\section{Supermaps and process matrices}

In this last section we give the converse proofs, which characterise theories which support processes matrix representation and furthermore equivalence between supermaps and process matrices. 
\begin{theorem}
Consider any generalised theory $\cat{C}$ with channel-state duality. The following are equivalent:
\begin{itemize}
\item Busch-Gleason axiom.
\item Every multipartite abstract process functional has a process-matrix representation in the following sense \[   \input{figs/process_thm_1.tikz}  \ = \   \input{figs/process_thm_2.tikz}    \]
\end{itemize}
\end{theorem}

\begin{proof}
We simply must show that the previously introduced non-signalling scenario is included within the multi-party instrument set. Indeed, examine the following \[   \input{figs/process_arrange_2.tikz} \ = \ \input{figs/process_arrange_3.tikz} \]  \[ = \  \input{figs/process_arrange_4.tikz} .   \] As a consequence, we can see by set inclusion $\mathcal{M}^2_{A \leq B} \subseteq \mathcal{M}^2$ that the non-signalling shape is also included within the domain for the definite-order process functionals. 

For the contrary direction, note that any family of functions \[    \input{figs/w_busch_1.tikz}   \] satisfying non-contextuality and commutativity on destructive measurements can be extended to an abstract process functional on destructive instruments in the following way \[    \input{figs/w_busch_3.tikz} \ := \   \input{figs/w_busch_4.tikz}.  \] This is indeed an abstract process functional, which can most easily be seen by checking commutativity:
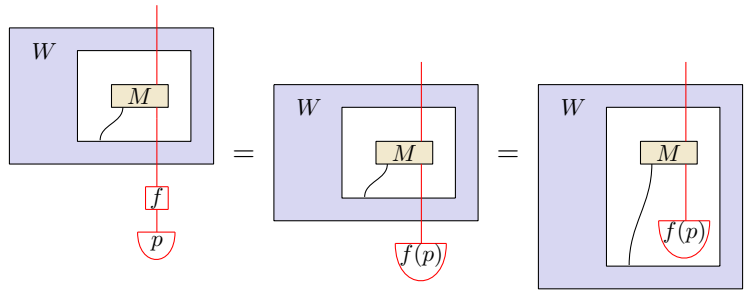
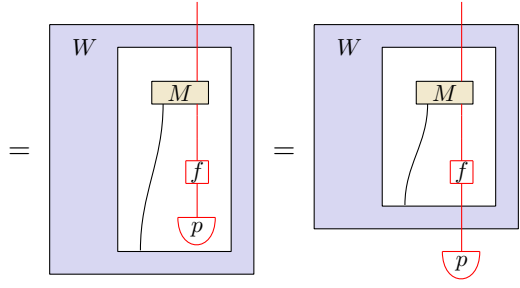
\[    \input{figs/w_busch_9.tikz} \ = \  \input{figs/w_busch_5.tikz} \ = \   \input{figs/w_busch_6.tikz}      \] 
\[  \ = \   \input{figs/w_busch_7.tikz} \ = \   \input{figs/w_busch_8.tikz}  .  \] 
Therefore, if there exists a process-matrix representation for this abstract process functional, the generalised theory at hand would have to satisfy the Busch-Gleason axiom.
\end{proof}
\begin{theorem}
Consider a generalised theory $\cat{C}$ with channel-state duality:
\[ \textrm{ $\cat{C}$ is tomographically closed}  \]
\[  \iff  \]
\[  \{ \textrm{Multipartite Process Matrices} \}  \cong  \{ \textrm{Supermaps of type $S: \bigtimes_i (A_i \Rightarrow A_i') \rightarrow (I \Rightarrow I)$} \}   \]
\end{theorem}
\begin{proof}
We stick to the $2$-party case for simplicity and note that constructing a process matrix from a supermap is immediate in any generalised theory with channel-state duality; the non-trivial direction is to construct a supermap from a process matrix. Let us consider the application of a process matrix to any element of the extended instrument set in which the right-hand process factors through the identity:
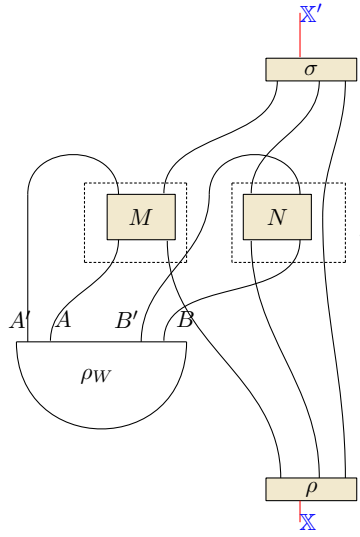
\[   \input{figs/stat_closed_2a.tikz} ,  \]
since $\rho_W$ is a process matrix, it follows that this classical channel is deterministic. Consequently, we find that $M_W$ defined by the dotted box of the following diagram satisfies the premise of Tomographic Closure \[   \input{figs/stat_closed_3a.tikz}.   \]
Since the generalised theory satisfies Tomographic Closure, it follows that for every $M,N$ the following diagram is a deterministic process \[   \input{figs/stat_closed_4a.tikz}.   \] As a result, up to isomorphism we can construct $S_W$ as denoted by the dotted box in the following picture, and see that indeed when applied to parts of arbitrary deterministic processes as in the following diagram, it returns a new deterministic process \[   \input{figs/stat_closed_5a.tikz}.   \] Consequently, $S_W$ is a supermap. 

For the contrary direction, note that for any process $E$ such that \[   \input{figs/stat_reverse_1.tikz}   \] is a deterministic classical channel, we can construct a candidate process matrix state $\rho_W$ via the following dotted box \[   \input{figs/stat_reverse_2.tikz}  \]. 
This state is indeed a two-party process matrix of type $(I \Rightarrow A') \otimes (B \Rightarrow I)$, since for any multiparty instrument  \[   \input{figs/stat_reverse_3.tikz}  \] we see that  \[   \input{figs/stat_reverse_4.tikz}  \]. Consequently, if process matrices are equal to supermaps then the process $E$ given by  \[   \input{figs/stat_reverse_5a.tikz} \ = \   \input{figs/stat_reverse_5b.tikz} \] is a supermap. Applying this supermap to the swap between the empty-space (monoidal unit) object and the input/output wires of the original process respectively, we find that \[  \input{figs/stat_reverse_5a.tikz} \ = \  \input{figs/stat_closed_6.tikz}  \] is by definition a deterministic process. This demonstrates that when process matrices are equal to supermaps, the Tomographic Closure condition must hold and so by the contrapositive, when Tomographic Closure fails, not every process matrix is a supermap.
\end{proof}

\end{document}

%% file: preamble.tex
\usepackage{amsmath}
\usepackage{amsthm}
\usepackage{amssymb}
\usepackage{physics}
\usepackage{stmaryrd}
\usepackage{tikz}
\usepackage{tikz-cd}
\usepackage{tikzit}
\usepackage{circuitikz}
\usepackage[numbers,sort&compress]{natbib}
\usepackage{hyperref}
\usepackage{tubes}
\usepackage{xsavebox}
\usepackage{cmll}

\input{styles.tikzdefs}
\input{styles.tikzstyles}
\usetikzlibrary{circuits.ee.IEC}

\newcommand{\tikzfigscale}[2]{\scalebox{#1}{\input{#2.tikz}}}

\newcommand{\cat}[1]{\mathcal{#1}}

\newcommand{\mat}[1]{\mathsf{Mat}_{#1}}
\newcommand{\opcat}[1]{#1^\textrm{op}}

\newcommand{\reals}{\mathbb{R}}

\newcommand{\stoch}{\mathsf{Stoch}}

\newcommand{\seq}{\olessthan}

\makeatletter
\newcommand{\xRightarrow}[2][]{\ext@arrow 0359\Rightarrowfill@{#1}{#2}}
\makeatother

%% file: styles.tikzstyles
\tikzstyle{discarding}=[fill=white, draw=black, shape=circle, style=upground]
\tikzstyle{smalldiscarding}=[fill=white, draw=black, style=upground, scale=0.75]
\tikzstyle{backdiscard}=[fill=white, draw=black, shape=circle, style=downground, scale=0.5]
\tikzstyle{smallbackdiscard}=[fill=white, draw=black, shape=circle, style=downground, scale=0.5]
\tikzstyle{state}=[fill=white, draw=black, style=triang, tikzit shape=rectangle]
\tikzstyle{kstate}=[fill=white, draw=black, style=kpoint, tikzit shape=rectangle]
\tikzstyle{kstateconj}=[fill=white, draw=black, style=kpoint conjugate, tikzit shape=rectangle]
\tikzstyle{kstateBIG}=[fill=white, draw=black, style=big kpoint, tikzit shape=rectangle]
\tikzstyle{effect}=[fill=white, draw=black, style=triangdag]
\tikzstyle{keffect}=[fill=white, draw=black, style=kpoint adjoint]
\tikzstyle{keffectconj}=[fill=white, draw=black, style=kpoint transpose]
\tikzstyle{morphdag}=[style=mapdag]
\tikzstyle{morph}=[style=hadamard]
\tikzstyle{WIDEmorph}=[style=hadamard, minimum width=14mm]
\tikzstyle{morphtrans}=[style=maptrans]
\tikzstyle{morphconj}=[style=mapconj]
\tikzstyle{CPMmorph}=[style=dmap]
\tikzstyle{CPMmorphconj}=[style=dmapconj]
\tikzstyle{CPMmorphdag}=[style=dmapdag]
\tikzstyle{CPMmorphtrans}=[style=dmaptrans]
\tikzstyle{CPMstate}=[fill=white, draw=black, style=triang, doubled]
\tikzstyle{CPMstateBIG}=[fill=white, draw=black, style={triang_lesssep}, doubled]
\tikzstyle{CPMkstate}=[fill=white, draw=black, style=kpoint, tikzit shape=rectangle, doubled]
\tikzstyle{CPMkstateconj}=[fill=white, draw=black, style=kpoint conjugate, tikzit shape=rectangle, doubled]
\tikzstyle{CPMkstateBIG}=[fill=white, draw=black, style=big kpoint, tikzit shape=rectangle, doubled]
\tikzstyle{CPMkeffect}=[fill=white, draw=black, style=kpoint adjoint, doubled]
\tikzstyle{CPMkeffectconj}=[fill=white, draw=black, style=kpoint transpose, doubled]
\tikzstyle{UHfB}=[fill=white, draw=black, style=triangdag, doubled, inner sep=-2pt]
\tikzstyle{leak}=[style=tinypoint, regular polygon rotate=-90]
\tikzstyle{leakfill}=[style=tinypoint, regular polygon rotate=-90, fill=black]
\tikzstyle{Z}=[style=dot, fill=green]
\tikzstyle{X}=[style=dot, fill=red]
\tikzstyle{black_dot}=[style=dot, fill=black]
\tikzstyle{white_dot}=[style=dot, fill=white]
\tikzstyle{qblack_dot}=[style=ddot, fill=black]
\tikzstyle{qwhite_dot}=[style=ddot, fill=white]
\tikzstyle{whitephase}=[style=wphase dot, fill=white]
\tikzstyle{qredphase}=[style=phase dot, fill=red]
\tikzstyle{qgreenphase}=[style=phase dot, fill=green]
\tikzstyle{had}=[style=hadamard, doubled]
\tikzstyle{box}=[style=hadamard]
\tikzstyle{bigbox}=[style=hadamard, minimum height=4mm, minimum width=8mm]
\tikzstyle{classhad}=[style=hadamard]
\tikzstyle{antipode}=[style=anti]

\tikzstyle{dottededge}=[-, dash pattern=on 1pt off 0.7pt]
\tikzstyle{double edge}=[-, style=doubled, draw=black, tikzit draw={rgb,255: red,18; green,168; blue,191}]
\tikzstyle{arrow}=[->]
\tikzstyle{new edge style 1}=[-, draw={rgb,255: red,242; green,233; blue,206}, fill={rgb,255: red,242; green,233; blue,206}]
\tikzstyle{morphism_shade}=[-, draw=black, fill={rgb,255: red,242; green,233; blue,206}, line join=bevel]
\tikzstyle{supermap_shade}=[-, fill={rgb,255: red,216; green,215; blue,242}, draw=black, line join=bevel]
\tikzstyle{hole_shade}=[-, fill=white, draw=black, line join=bevel]
\tikzstyle{new edge style 2}=[-, draw={rgb,255: red,14; green,188; blue,83}]
\tikzstyle{new edge style 3}=[<-, draw={rgb,255: red,234; green,209; blue,255}]
\tikzstyle{new edge style 4}=[<-, draw={rgb,255: red,255; green,0; blue,0}]
\tikzstyle{new edge style 5}=[-, draw={rgb,255: red,214; green,110; blue,62}]
\tikzstyle{new edge style 6}=[-, draw={rgb,255: red,174; green,20; blue,174}]
\tikzstyle{new edge style 0}=[-, fill=none, draw={rgb,255: red,255; green,0; blue,0}]
\tikzstyle{green_dotted}=[-, draw={rgb,255: red,0; green,106; blue,106}, dash pattern=on 1pt off 0.7pt]

%% file: figs/snake_2.tikz
\begin{tikzpicture}[scale=0.6, {every node/.style}={scale=0.8}]
	\begin{pgfonlayer}{nodelayer}
		\node [style=none] (2) at (-0.75, -1) {};
		\node [style=none] (3) at (0.75, 1) {};
	\end{pgfonlayer}
	\begin{pgfonlayer}{edgelayer}
		\draw [in=270, out=90] (2.center) to (3.center);
	\end{pgfonlayer}
\end{tikzpicture}